\documentclass[onecollarge]{svjour2}       
\smartqed  
\usepackage{amsfonts}
\usepackage{amssymb}
\usepackage{amsmath}
\usepackage{mathpazo}
\usepackage{hyperref}
\usepackage{multimedia}
\usepackage{epsfig}
\usepackage{xcolor}
\usepackage{amssymb}
\usepackage{graphicx}
\usepackage{url}

\begin{document}

\title{Optimal Codes from Fibonacci Polynomials and Secret Sharing Schemes}



\author{Mehmet E. Koroglu \and Ibrahim Ozbek \and Irfan Siap}


\institute{Mehmet E. Koroglu \at
              Department of Mathematics, Yildiz Technical University, Esenler, Istanbul 34220, Turkey\\
              \email{mkoroglu@yildiz.edu.tr}
              \and
         Ibrahim Ozbek \at
              Department of Mathematics, Yildiz Technical University, Esenler, Istanbul 34220, Turkey\\
              \email{ibrhmozbk@gmail.com}
              \and
         Irfan Siap \at
              Department of Mathematics, Yildiz Technical University, Esenler, Istanbul 34220, Turkey \\
              \email{irfan.siap@gmail.com}}

\date{Received: date / Accepted: date}

\maketitle

\begin{abstract}
In this work, we study cyclic codes that have generators as Fibonacci
polynomials over finite fields. We show that these cyclic codes in most
cases produce families of maximum distance separable and optimal codes with
interesting properties. We explore these relations and present some
examples. Also, we present applications of these codes to secret sharing
schemes.
\end{abstract}

\section{Introduction}
Error correcting codes are applied intensively in digital data transfer and
storage. Due to this important nature, good error correcting codes which can
be considered as codes with best possible parameters so called optimal codes
and codes with rich algebraic structures are important for implementations.
Cyclic codes serve such a purpose and studies on cyclic codes still cover an
important part of the area. A linear code $C$ is a subspace of $V=\mathbb{F}%
_{p}^{n}$ where $\mathbb{F}_{p}$ denotes the finite field with $p$ elements.
The elements of a linear code are called codewords. In order to detect hence
correct errors, Hamming metric serves such a purpose. Given two elements in $%
V$ say $u=(u_{1},u_{2},\ldots ,u_{n}),$ $v=(v_{1},v_{2},\ldots ,v_{n}),$ the
Hamming distance between $u$ and $v$ is the number of places that differ
from each other i.e. $d(u,v)=|\{i|u_{i}\neq v_{i}\}|.$ The smallest nonzero
Hamming distance among the elements of $C$ is referred as the Hamming
distance of the code $C$ and it is usually denoted by $d(C)$ or simply $d.$
If $C$ is a linear code over $\mathbb{F}_{p}$ with dimension $k$ and minimum
distance $d,$ then $C$ is said to be an $[n,k,d]_{p}$-code. If an error
detected while the received word is erroneous, then decoding this word to
the closest codeword in $C$ is called the majority decoding method. It is
also well known that a linear code with minimum distance $d=2t+1$ or $d=2t+2$
can correct up to $t$ errors. Given the length and the dimension, finding a
linear code with best possible minimum distance is an important problem and
it is an open problem except a few special cases. An inner product on $V$ of
$u=(u_{1},u_{2},\ldots ,u_{n}),$ $v=(v_{1},v_{2},\ldots ,v_{n}),$ is defined
as usual $\langle u,v\rangle =\sum_{i=1}^{n}u_{i}v_{i}$ in $\mathbb{F}_{p}.$
Then, we can associate a linear code $C^{\perp },$ called the dual code of $%
C $, to a code $C$ by $C^{\perp }=\{v\in V|\langle u,v\rangle =0\text{ for
all }u\in C\}.$ If $C$ is a linear code of length $n$ and dimension $k,$
then it is well-known that $C^{\perp }$ is a linear code of length $n$ and
dimension $n-k$. In literature, cyclic codes from sequences defined over
some extension fields with special generators have been studied. Sequences
over fields or rings in general have many applications such as Left Shift
Registers (LSR), coding, cryptography, etc \cite{crypt},\cite{crypt2}. This venue
of the research is partially accomplished by considering some special
sequences. In each case some special sequences are studied in order to
understand the cyclic codes derived from them. Here, we study cyclic codes
derived from Fibonacci sequences. In the literature there are studies where
Fibonacci sequences and codes are related but to the best knowledge of the
authors these studies are in different directions compared to the one
presented in this paper. An example of such study is done by Lee et al. \cite%
{fibo1} where linear codes related to Fibonacci sequences are presented and
burst error correction of such families are studied. There are further
studies that are inspired by Fibonacci sequences \cite{fibo3},\cite{fibo2}. In the
sequel we present some basic properties of both Fibonacci sequences and
error correcting codes. In the next section, we relate Fibonacci sequences
with cyclic codes and we study the properties of such codes. Moreover, we
apply these families of codes to construction of secret sharing schemes.

\subsection{Fibonacci Sequences and Some Properties}

In this subsection, we present some basic properties and theorems regarding
Fibonacci sequences that are going to be useful in the following sections.
\begin{definition}
Let $F_{0}=0$ and $F_{1}=1$ be elements of a finite field $\mathbb{F}_{p}.$
Then, the sequence defined by $F_{n}=F_{n-1}+F_{n-2}$ for $n\geq 2$ is
called the Fibonacci sequence in $\mathbb{F}_{p}.$ If the we take the first
two terms of the sequence as $F_{0}=a$ and $F_{1}=b$, then the sequence is
called generalized Fibonacci sequence and we will denote generalized
Fibonacci sequence by $\bar{G}\left( a,b\right)$.
\end{definition}%
\begin{definition}
The smallest $t>0$ such that $F_{0}\equiv F_{t}\mod p$ and $F_{1}\equiv
F_{t+1}\mod p,$ where $F_{t}$ is the $t^{th}$ Fibonacci number, is called
Pisano period of $p$ and we denote this period by $l_{p}.$
\end{definition}For example, the Fibonacci sequence computed in $\mathbb{F}%
_{11}$ is
\begin{equation*}
\underline{0,1,1,2,3,5,8,2,10,1},\mathbf{0,1,1}\ldots
\end{equation*}
which implies that Pisano period $l_{11}$ is $10.$
\begin{theorem}
\label{wall1} \cite{wall} $F_{n}\mod p$ forms a periodic sequence. That is,
the sequence keeps repeating its values periodically.
\end{theorem}There has not been established a direct formula for computing $%
l_{p}$ yet. However, the following theorem gives a restriction for possible
values of $l_{p}$.
\begin{theorem}
\label{wall2} \cite{wall} Let $l_{p}$ denote the period of the Fibonacci
sequence modulo $p.$ Then,
\begin{enumerate}
\item If $p$ is prime and $p\equiv \pm 1\mod 10,$ then $l_{p}|p-1$.
\item If $p$ is prime and $p\equiv \pm 3\mod 10,$ then $l_{p}|2(p+1)$.
\end{enumerate}
\end{theorem}%
\begin{lemma}
\cite{Burton} $5$ is a quadratic residue modulo primes of the form $5t\pm 1$
and a quadratic non-residue modulo primes of the form $5t\pm 2$.
\end{lemma}%
\begin{lemma}
\cite{vajda} \label{vajda1} If a prime $p$ is of the form $5t\pm 1$ then $%
F_{p-1}\equiv 0$ and $F_{p}\equiv 1\mod p$. If a prime p is of the form $%
5t\pm 2,$ then $F_{p}\equiv -1$ and $F_{p+1}\equiv 0\mod p$.
\end{lemma}Let $\alpha \left( p\right) $ denote the index of the subscript
of the first nonzero term of the Fibonacci sequence which is divisible by $p$%
. Let $s\left( p\right) $ be the least residue of $F_{\alpha \left( p\right)
+1}\mod
p$ and let $\beta \left( p\right) $ denote the order of $s\left( p\right) $
modulo $p$ i.e. the smallest positive integer $\beta \left( p\right) $ such
that $s\left( p\right) ^{\beta \left( p\right) }\equiv 1\mod p.$
\begin{table}
\caption{The table indicates the connection between Pisano period $\left(
l_{p}\right) $, the index of the first term of the Fibonacci sequence which
is zero $\left( \protect\alpha \left( p\right) \right) $, the least residue
of sequence $\left( s\left( p\right) \right) $, and the number of zeros in a
single period of Fibonacci sequence computed in $\mathbb{F}_{p}$ (or
equivalently the order of $s\left( p\right) ,\left( \protect\beta \left(
p\right) \right) $ modulo prime $p$).}
\label{tab:1}
\begin{center}
\begin{tabular}{cccccl}
\hline
$p$ & $l_{p}$ & $\alpha \left( p\right) $ & $s\left( p\right) $ & $\beta
\left( p\right) $ & The corresponding Fibonacci sequence \\ \hline
$7$ & $16$ & $8$ & $6$ & $2$ & $\left\{
0,1,1,2,3,5,1,6,0,6,6,5,4,2,6,1\right\} $ \\
$11$ & $10$ & $10$ & $1$ & $1$ & $\left\{ 0,1,1,2,3,5,8,2,10,1\right\} $ \\
$13$ & $28$ & $7$ & $8$ & $4$ & $\left\{ 0,1,1,2,3,5,8,0,8,8,\ldots
,2,12,1\right\} $ \\
$17$ & $36$ & $9$ & $4$ & $4$ & $\left\{ 0,1,1,2,3,5,8,13,4,0,4,\ldots
,2,16,1\right\} $ \\
$19$ & $18$ & $18$ & $1$ & $1$ & $\left\{
0,1,1,2,3,5,8,13,2,15,17,13,11,5,16,2,18,1\right\} $ \\
$23$ & $48$ & $24$ & $22$ & $2$ & $\left\{ 0,1,1,2,3,5,8,13,21,\ldots
,2,22,1\right\} $ \\ \hline
\end{tabular}%
\end{center}
\end{table}

\begin{theorem}
\cite{rob} \label{rob1} $l_{p}=\alpha \left( p\right) \beta \left( p\right).$
\end{theorem}As a result of Theorem \ref{rob1}, $\beta \left( p\right) $ can
be considered as the number of zeros in a single period of Fibonacci
sequence computed in $\mathbb{F}_{p}$.
\begin{theorem}
\cite{rob} \label{rob2} $l_{p}=\gcd \left( 2,\beta \left( p\right) \right) .%
\text{ lcm } \left[ \alpha \left( p\right) ,\gamma \left( p\right) \right] $%
, where $\gamma \left( 2\right) =1$ and $\gamma \left( p\right) =2$ for $%
p>2. $
\end{theorem}%
\begin{corollary}
\label{rob3} \cite{rob}
\begin{enumerate}
\item $l_{p}$ is even for $p>2.$
\item $\beta \left( p\right) =1,2,$ or $4.$
\end{enumerate}
\end{corollary}

\subsection{Error Correcting Code Basics}

In this subsection we present some basic theory about linear codes
especially cyclic codes. For further and more detailed information regarding
this topic the readers may refer to \cite{Ling}. Let $\mathbb{F}_{p}$ be a
finite field with $p$ elements where $p$ is prime and $\mathbb{F}_{p}^{n}$
be an $n$ dimensional vector space. A linear code $C$ of length $n$ over $%
\mathbb{F}_{p}$ is a subspace of $\mathbb{F}_{p}^{n}.$ A subset $S$ of $%
\mathbb{F}_{p}^{n}$ is cyclic if
\begin{equation*}
(a_{n-1},a_{0},a_{1},\ldots ,a_{n-2})\in S\text{ whenever }%
(a_{0},a_{1},\ldots ,a_{n-1})\in S.
\end{equation*}%
A linear code $C$ is called a cyclic code if $C$ is also a cyclic set.
\begin{definition}
\cite{Ling} Let $\alpha $ be a primitive element of $\mathbb{F}_{p^{m}}$ and
denote by $M^{(i)}(x)$ the minimal polynomial of $\alpha ^{i}$ with respect
to $\mathbb{F}_{p}.$ A (primitive) $BCH$ code over $\mathbb{F}_{p}$ of
length $n=p^{m}-1$ with designed distance $\delta $\ is a $p-$ary cyclic
code generated by
\[
g(x):=lcm(M^{(a)}(x),M^{(a+1)}(x),...,M^{(a+\delta -2)}(x))
\]
for some integer $a.$ Furthermore, the code is called narrow-sense if $a=1.$
\end{definition}%
\begin{definition}
\cite{Ling} A $p-$ary Reed Solomon code ($RS$ code) is a $p-$ary $BCH$ code
of length $p-1$ generated by
\[
g(x)=(x-\alpha ^{a+1})(x-\alpha ^{a+2})\ldots (x-\alpha ^{a+\delta -1}),
\]%
with $a\geq 0$ and $2\leq \delta \leq p-1$, where $\alpha $ is a primitive
element of $\mathbb{F}_{p}$.
\end{definition}In order to relate the combinatorial structure of cyclic
codes with algebraic structures, the following map $\varphi $ is defined as
\begin{eqnarray}
\varphi &:&\mathbb{F}_{p}^{n}\rightarrow \mathbb{F}_{p}[x]/(x^{n}-1)  \notag
\\
\varphi \left( a_{0},a_{1},\ldots ,a_{n-2},a_{n-1}\right)
&=&a_{0}+a_{1}x+\ldots +a_{n-2}x^{n-2}+a_{n-1}x^{n-1}.  \label{A}
\end{eqnarray}%
$\mathbb{F}_{p}[x]/(x^{n}-1)$ is a principal ideal rings and $\varphi $
corresponds each subspace of $\mathbb{F}_{p}^{n}$ to an ideal in $\mathbb{F}%
_{p}[x]/(x^{n}-1).$ The following theorem states this connection:
\begin{theorem}
\cite{Ling} Let $\varphi $ be the linear map defined in the Equation %
\eqref{A}. Then a nonempty subset $C$ of $\mathbb{F}_{p}^{n}$ is a cyclic
code if and only if $\varphi (C)$ is an ideal of $\mathbb{F}%
_{p}[x]/(x^{n}-1) $.
\end{theorem}%
\begin{corollary}
\label{c2} \cite{Ling} The nonempty subset $C$ of $\mathbb{F}_{p}^{n}$ is a
cyclic code of length $n$ if and only if $g(x)|x^{n}-1$ and $\varphi
(C)=\left\langle g(x)\right\rangle $.
\end{corollary}%
\begin{example}
\label{B} The code $C=\left\{ {000,111,222}\right\} $ is a ternary cyclic
code. The corresponding ideal in $\mathbb{F}_{3}[x]/(x^{3}-1)$ is $\varphi
(C)=\left\{ 0,1+x+x^{2},2+2x+2x^{2}\right\} =\left\langle {1+x+x^{2}}%
\right\rangle .$
\end{example}Since the minimum distance determines the error correction and
detection capability of a code, it is an important parameter for codes, and
also determining it is a very difficult problem. There are at least some
bounds that help on estimating the minimum distance of a code. Now, we
present definitions and theorems regarding some special bounds that are
going to be referred in the sequel.
\begin{definition}
\cite{Ling} [Singleton bound] \label{single} If $C$ is a linear code with
parameters $[n, k, d],$ then $k\leq n-d+1.$
\end{definition}%
\begin{definition}
\cite{Ling} \label{mds} A linear code with parameters $[n, k, d]$ such that $%
k+d = n+1$ is called a maximum distance separable (MDS) code.
\end{definition}The code presented in Example \ref{B} is MDS.
\begin{theorem}[Griesmer Bound]
\cite{Ling} Let $C$ be a $p-$ary code of parameters $[n,k,d]$, where $k\geq
1 $. Then%
\[
n\geq \sum\limits_{i=0}^{k-1}\left\lceil \frac{d}{p^{i}}\right\rceil .
\]%
Here, if $\alpha $ is a real number, then $\left\lceil \alpha \right\rceil $
denotes the smallest integer larger or equal to (the ceil) $\alpha.$
\end{theorem}%
\begin{example}
\label{C} Let $C$ be a cyclic linear code generated by the polynomial $%
g\left( x\right)
=x^{14}+6x^{13}+2x^{12}+4x^{11}+5x^{10}+6x^{9}+6x^{8}+6x^{6}+x^{5}+5x^{4}+3x^{3}+2x^{2}+x+1
$ over $\mathbb{F}_{7}[x]/(x^{16}-1)$. The code $C$ has parameters $%
[16,2,14]_{7}$. Thus, we have $16\geq \sum\limits_{i=0}^{1}\left\lceil \frac{%
14}{7^{i}}\right\rceil =14+2=16.$ So, the given code meets the Griesmer
bound.
\end{example}

\section{Cyclic Codes Obtained from Fibonacci Polynomials}

In this section we study cyclic codes that are generated by polynomials
related to Fibonacci sequences.
\begin{definition}
\cite{Alg} Let $S=\left\{ s_{0},s_{1},\ldots \right\} $ be an arbitrary
sequence over $\mathbb{F}_{p}.$ Assume that $s_{i+n}+c_{n-1}s_{i+n-1}+\ldots
c_{0}s_{i}=0$ for some elements $c_{0},c_{1},\ldots ,c_{n-1}\in \mathbb{F}%
_{p}$ and for all $i=0,1,\ldots .$ Then the polynomial $%
x^{n}+c_{n-1}x^{n-1}+\ldots +c_{1}x+c_{0}$ is called the characteristic
polynomial of $S.$ A characteristic polynomial of minimal degree is called
the minimal polynomial of $S.$
\end{definition}
For a periodic sequence $S=\left\{ s_{0},s_{1},\ldots \right\} $ with a
period $N$ we have $s_{i+N}-s_{i}=0,$ so $x^{N}-1$ is a characteristic
polynomial of $S.$ If $N$ is a period of $S,$ then the minimal polynomial of
$S$ is%
\begin{equation*}
\frac{x^{N}-1}{\gcd \left( x^{N}-1,s_{N-1}x^{N-1}+\ldots
+s_{1}x+s_{0}\right) }.
\end{equation*}
Let $F=\left\{ F_{0},F_{1},\ldots ,F_{n},\ldots \right\} $ denote the
Fibonacci sequence over $\mathbb{F}_{p}$ and suppose that the period of this
sequence is equal to $l.$ The polynomial $f(x)=\sum_{i=0}^{l-1}F_{i}x^{i}\in
\mathbb{F}_{p}[x]$ is called Fibonacci polynomial of $F$ over $\mathbb{F}%
_{p}.$
\begin{theorem}
\label{siap} Let $f(x)\in \mathbb{F}_{p}[x]$ be the Fibonacci polynomial
with period $l=p-1$ and $\beta \left( p\right) =1.$ Then,
\begin{enumerate}
\item $(f(x),x^{p-1}-1)=\frac{x^{p-1}-1}{x^{2}+x-1}\in \mathbb{F}_{p}[x]$
\item The cyclic code $C=\langle f(x)\rangle $ generated by $f(x)$ with
dimension $2,$ and the minimum distance $d=p-2.$
\item $C=\langle f(x)\rangle $ is an MDS code of type $[p-1,2,p-2]_{p}.$
\end{enumerate}
\end{theorem}
\begin{proof}
\begin{enumerate}
\item By direct checking and applying the properties of Fibonacci sequence
we see that $f(x)(x^{2}+x-1)=x^{l}-x\in \mathbb{F}_{p}[x].$ Since $x|f(x)$
but $x{\nmid }x^{p-1}-1$ we have $(f(x),x^{p-1}-1)=\frac{x^{p-1}-1}{x^{2}+x-1}.$
\item Let $g(x)=(f(x),x^{p-1}-1)=\frac{x^{p-1}-1}{x^{2}+x-1}$ be the cyclic
code of length $p-1$ generated by $g(x)$ and dimension $p-1-\deg (g(x)))=2.$
This codes has exactly $p^{2}$ codewords. Since $g(x)=\frac{f(x)}{x},$ $%
w(f(x))=w(g(x)).$ Since the number of zeros in a single period of given
Fibonacci sequence is only one, then we have $w(g(x))=p-2.$ Also, $%
w(xg(x))=p-2$ for which $xg(x)\in C.$ The codeword $g(x)$ has the zero entry
in its first coordinate and $xg(x)$ has the zero entry in its second
coordinate. Suppose that a codeword $c(x)\in C$ has two entries with zeros.
Then, ${g(x),xg(x),c(x)}$ will give a linearly independent subset of vectors
in $C$ which is a contradiction to the dimension $2$ of $C.$
\item Follows from previous part and Definition \ref{mds}.
\end{enumerate}
\end{proof}
\begin{corollary}
\label{c3} The codes given in Theorem \ref{siap} are $RS$ codes.
\end{corollary}
\begin{theorem}
\label{siap4} Let $f(x)\in \mathbb{F}_{p}[x]$ be the Fibonacci polynomial
with period $l=p-1,$ $\beta \left( p\right) =1$ and $C=\langle f(x)\rangle $%
. Then, the dual code of $C,$ $C^{\bot }=\left\langle x^{2}+x-1\right\rangle
$ is an MDS code with parameters $[p-1,p-3,3]_{p}$.
\end{theorem}
\begin{proof}
Clearly the length of $C^{\bot }$ is $p-1.$ Let us now
determine the dimension and minimum distance of $C^{\bot }.$ Since $C^{\bot
}=\langle x^{2}+x-1\rangle $, we have dim$\left( C^{\bot }\right)
=p-1-2=p-3. $ We know that $C$ is MDS, so is $C^{\bot }$\cite{MacW}. By the
Singleton bound we have $p-1+1=p-3+d.$ Thus $d=3.$
\end{proof}
\begin{corollary}
\label{siap1} Let $f(x)\in \mathbb{F}_{p}[x]$ be the Fibonacci polynomial
with period $l=p-1.$
\begin{enumerate}
\item If $\beta \left( p\right) =2,$ then $C=\langle f(x)\rangle $ is a
cyclic code of type $[p-1,2,p-3]_{p}.$
\item If $\beta \left( p\right) =4,$ then $C=\langle f(x)\rangle $ is a
cyclic code of type $[p-1,2,p-5]_{p}.$
\end{enumerate}
\end{corollary}
\begin{theorem}
\label{siap2} Let $f(x)\in \mathbb{F}_{p}[x]$ be the Fibonacci polynomial
with period $l=2p+2$ and $\beta \left( p\right) =2,4.$ Then,
\begin{enumerate}
\item $(f(x),x^{p-1}-1)=\frac{x^{p-1}-1}{x^{2}+x-1}\in \mathbb{F}_{p}[x]$
\item The cyclic code $C=\langle f(x)\rangle $ generated by $f(x)$ has
dimension $2,$ and the minimum distance $d=2p+2-\beta \left( p\right) $.
\item $C=\langle f(x)\rangle $ is a linear code of type $[2p+2,2,2p+2-\beta
\left( p\right) ]_{p}.$
\end{enumerate}
\end{theorem}
\begin{proof}
Follows from Theorem \ref{siap}.
\end{proof}
\begin{corollary}
\label{siap3} Let $f(x)\in \mathbb{F}_{p}[x]$ be the Fibonacci polynomial
with period $l=2p+2.$
\begin{enumerate}
\item If $\beta \left( p\right) =2,$ then $C=\langle f(x)\rangle $ is an
optimal code of type $[2p+2,2,2p]_{p}.$
\item If $\beta \left( p\right) =4,$ then $C=\langle f(x)\rangle $ is a code
of type $[2p+2,2,2p-2]_{p}.$
\end{enumerate}
\end{corollary}
\begin{proof}
\begin{enumerate}
\item Since $\beta \left( p\right) =2$, by the Theorem \ref{siap2}, $%
C=\langle f(x)\rangle $ is a $[2p+2,2,2p]_{p}$-code. Recall that the
Griesmer bound for a linear code $C$ of type $\left[ n,k,d\right] _{p}$ is $%
n\geqslant \sum\limits_{i=0}^{k-1}\left\lceil \frac{d}{p^{i}}\right\rceil .$
Thus we have $2p+2\geqslant \left\lceil \frac{2p}{p^{0}}\right\rceil
+\left\lceil \frac{2p}{p^{1}}\right\rceil =2p+2.$ So, $C=\langle f(x)\rangle
$ is an optimal code.
\item Clearly, if $\beta \left( p\right) =4,$ then by the Theorem \ref{siap2}%
, $C=\langle f(x)\rangle $ is a code of type $[2p+2,2,2p-2]_{p}.$
\end{enumerate}
\end{proof}
\begin{lemma}
\label{siap5} Let $F_{n}$ be a Fibonacci sequences with period $l$ over $%
\mathbb{F}_{p}.$Then, there are $\frac{l-\beta \left( p\right) }{\beta
\left( p\right) }-1$ non-zero, non-multiple and different, two consecutive
terms in a sequence.
\end{lemma}
\begin{proof}
We know that there are $l-\beta \left( p\right) $ nonzero
terms in a sequence and if we obtain zero term in a sequence, then the
length of subsequent terms until another zero term must be a multiple of the
preceding part of the sequence. So there are $\frac{l-\beta \left( p\right)
}{\beta \left( p\right) }-1$ non-zero, non-multiple and different, two
consecutive terms in a sequence.
\end{proof}
\begin{theorem}
\label{c6} The cyclic codes given in the Theorem \ref{siap}, Corollary \ref%
{siap1} and Theorem \ref{siap2} are constant one or two weight codes.
\end{theorem}
\begin{proof}
Since the dimension of given codes are $2,$ the generator
matrix is of the form%
\[
G=\binom{f(x)}{xf(x)}.
\]%
Thus we can generate in total $p^{2}$ codewords from the generator matrix $G$
where clearly one of them is all the zero vector. The number of nonzero
coefficients of $\ f(x)$ is the same as with $xf(x)$ and the weight of these
codewords are $l-\beta \left( p\right) .$ So,\ from the rows of generator
matrix $G,$ we can obtain $2\left( p-1\right) $ codewords of weight $l-\beta
\left( p\right) .$ Also, generalized Fibonacci sequences (for details see
\cite{Horadam}) satisfies the following facts:
\begin{enumerate}
\item If we start with any two consecutive fibonacci numbers for $a$ and $b$%
, $\bar{G}(a,b,i)$ will be essentially the same as the fibonacci sequence
but with its indices changed. The general rule is%
\begin{equation}
\bar{G}(f\left( k\right) ,f\left( k+1\right) ,i)=f\left( i+k\right)
\label{gf}
\end{equation}
\item Multiplying all the terms by $k$ gives the same sequence as the one
with starting values $ka$ and $kb$%
\begin{equation}
\bar{G}(ka,kb,i)=k\bar{G}(a,b,i).  \label{gf1}
\end{equation}
\end{enumerate}
The Equation \eqref{gf} says that if we start with any two consecutive
Fibonacci numbers, then we obtain codewords whose weight are the same with
the codewords of the form $f(x).$ By Lemma \ref{siap5}, there are $\frac{%
l-\beta \left( p\right) }{\beta \left( p\right) }-1$ non-zero, non-multiple
and different, two consecutive terms in a Fibonacci sequence$\mod
p$. Also from the Equation \eqref {gf1}, there are $p-1$ multiple of each
non-zero two consecutive terms. Then, we have $\left( \frac{l-\beta \left(
p\right) }{\beta \left( p\right) }-1\right) \left( p-1\right) $ codewords
which has same weight with $f(x).$ Therefore, the total number of the
codewords of weight $l-\beta \left( p\right) $ are $2\left( p-1\right)
+\left( \frac{l-\beta \left( p\right) }{\beta \left( p\right) }-1\right)
\left( p-1\right) .$ Since the minimum weight of the codes given in Theorem %
\ref{siap}, Corollary \ref{siap1} and Theorem \ref{siap2} are $d=$ $l-\beta
\left( p\right) $, the other weights of the codewords must be $l$. Because\
every Fibonacci sequence $\mod p$ has $\beta \left( p\right) $ parts and all
these parts have the same weight. Thus, the number of codewords of weight $l$
are
\[p^{2}-\left( 1+2\left( p-1\right) +\left( \frac{l-\beta \left( p\right)
}{\beta \left( p\right) }-1\right) \left( p-1\right) \right).\]
\end{proof}
\bigskip
\begin{corollary}
\label{c7} Let $A_{w}$ denote the number of codewords with Hamming weight $w$
in the codes given in the Theorem \ref{siap}, Corollary \ref{siap1} and
Theorem \ref{siap2}. Then the weight distribution of these codes are given
in Table \ref{tab:2}, \ref{tab:3}, \ref{tab:4}.
\end{corollary}
\begin{table}
\caption{The weight distribution of the codes given in Theorem \protect\ref%
{siap} for $\protect\beta \left( p\right) =1$}
\label{tab:2}
\begin{center}
\begin{tabular}{ccc}
\hline
Length & Weight $w$ & Multiplicity $A_{w}$ \\ \hline
$l=p-1$ &
\begin{tabular}{c}
$0$ \\
$p-2$ \\
$p-1$%
\end{tabular}
&
\begin{tabular}{c}
$1$ \\
$\left( p-1\right) ^{2}$ \\
$2\left( p-1\right)$%
\end{tabular}
\\ \hline
\end{tabular}%
\end{center}
\end{table}

\begin{table}
\caption{The weight distribution of the codes given in Corollary \protect\ref%
{siap1} and Theorem \protect\ref{siap2} for $\protect\beta \left( p\right)
=2 $}
\label{tab:3}
\begin{center}
\begin{tabular}{ccc}
\hline
Length & Weight $w$ & Multiplicity $A_{w}$ \\ \hline
$l=p-1$ &
\begin{tabular}{c}
$0$ \\
$p-3$ \\
$p-1$%
\end{tabular}
&
\begin{tabular}{c}
$1$ \\
$\frac{\left( p-1\right) ^{2}}{2}$ \\
$\frac{\left( p-1\right) \left( p+3\right) }{2}$%
\end{tabular}
\\ \hline
$l=2p+2$ &
\begin{tabular}{c}
$0$ \\
$2p$%
\end{tabular}
&
\begin{tabular}{c}
$1$ \\
$p^{2}-1$%
\end{tabular}
\\ \hline
\end{tabular}%
\end{center}
\end{table}
\begin{table}
\caption{The weight distribution of the codes given in Corollary \protect\ref%
{siap1} and Theorem \protect\ref{siap2} for $\protect\beta \left( p\right)
=4 $}
\label{tab:4}
\begin{center}
\begin{tabular}{ccc}
\hline
Length & Weight $w$ & Multiplicity $A_{w}$ \\ \hline
$l=p-1$ &
\begin{tabular}{c}
$0$ \\
$p-5$ \\
$p-1$%
\end{tabular}
&
\begin{tabular}{c}
$1$ \\
$\frac{\left( p-1\right) ^{2}}{4}$ \\
$\frac{\left( p-1\right) \left( 3p+5\right) }{4}$%
\end{tabular}
\\ \hline
$l=2p+2$ &
\begin{tabular}{c}
$0$ \\
$2p-2$ \\
$2p+2$%
\end{tabular}
&
\begin{tabular}{c}
$1$ \\
$\frac{p^{2}-1}{2}$ \\
$\frac{p^{2}-1}{2}$%
\end{tabular}
\\ \hline
\end{tabular}%
\end{center}
\end{table}

\subsection{Cyclic Codes From Generalized Fibonacci Polynomials}

If we take $F_{0}=2$ and $F_{1}=1$ as initial values for Fibonacci sequence,
then we obtain the well-known Lucas sequence. A natural problem then is to
figure out the structure of cyclic codes if they are defined in a more
general setting, with general initial values, which is known as generalized
Fibonacci sequences. Let $f(x)$ be a generalized Fibonacci polynomial with $%
F_{0}=a$ and $F_{1}=b.$ Then,
\begin{equation*}
f(x)(x^{2}+x-1)=ax^{p+1}+ax-a+bx^{p}-bx.
\end{equation*}%
If we consider this polynomial as the generator polynomial of a cyclic code,
as generator polynomial of an ideal in $\mathbb{F}_{p}[x]/(x^{p-1}-1),$ then
this means that
\begin{equation*}
f(x)(x^{2}+x-1)=ax^{2}+ax-a+bx-bx=a(x^{2}+x-1)\text{\ }\mod x^{p-1}-1.
\end{equation*}%
Hence, if $a\neq 0,$ then $a$ is a unit in the ring $R=\mathbb{F}%
_{p}[x]/(x^{p-1}-1)$ (or equivalently $\left( f(x),x^{p-1}-1\right) =1$) and
thus
\begin{equation*}
\langle f(x)\rangle =R.
\end{equation*}%
So, the only case to get a non trivial cyclic code from generalized
Fibonacci polynomials are the cases where $a=0.$ In case where $a=0$ all
cyclic codes (ideals) are the same. So the Fibonacci polynomial codes are
the only interesting ones among this family that we already studied above.

\subsection{Cyclic Codes From Extended Fibonacci Polynomials}

Let $E_{0}=0$ and $E_{1}=\ldots =E_{r-1}=1$ be elements of a finite field $%
\mathbb{F}_{p}.$ Then, the sequence defined by $E_{n}=\sum%
\limits_{j=n-r}^{n-1}E_{j}$ for $n\geq r$ we called the Extended Fibonacci
sequence in $\mathbb{F}_{p}.$ For example, for $r=3$ the Extended Fibonacci
sequence computed in $\mathbb{F}_{3}$ is
\begin{equation*}
\underline{0,1,1,2,1,1,1,0,2,0,2},1,0,0,\mathbf{1,1,}\ldots
\end{equation*}
Let $E=\left\{ E_{0},E_{1},\ldots ,E_{l},\ldots \right\} $ denote the Extended
Fibonacci sequence over $\mathbb{F}_{p}$ and suppose that the period of this
sequence is equal to $l.$ The polynomial $t(x)=\sum_{i=0}^{l-1}E_{i}x^{i}\in
\mathbb{F}_{p}[x]$ is called Extended Fibonacci polynomial of $E$ over $\mathbb{F}%
_{p}.$

\begin{theorem}
\label{t}
Let $(p=7,13)$ and $t(x)\in \mathbb{F}_{p}[x]$ be the extended Fibonacci polynomial
with period $l=p^{r-1}-1$ and $\beta\left( p\right)$ be the number of zeros of given extended fibonacci sequence. Then,
the cyclic code $C=\langle t(x)\rangle $ generated by $t(x)$ has
dimension $r,$ and the minimum distance $d=p^{r-1}-1-\beta \left( p\right)$.
So, $C=\langle t(x)\rangle $ is a linear code of type $[p^{r-1}-1,r,p^{r-1}-1-\beta
\left( p\right) ]_{p}.$
\end{theorem}
\begin{table}
\caption{Extended Fibonacci sequences computed in $\mathbb{F}_{7}$ and $\mathbb{F}_{13}.$}
\label{tab:5}
\begin{center}
\begin{tabular}{cccc}
\hline
$p$ & $l_{p}$ & $\beta\left( p\right) $ & The corresponding extended Fibonacci sequence \\ \hline
$7$ & $48$ & $12$ & $\left\{0,1,1,2,4,0,6,3,2,4,2,\ldots,2,2,4,1,0,5,6,4,1,4,2,0,6,1,0\right\} $ \\
$13$ & $168$ & $18$ & $\left\{ 0,1,1,2,4,7,0,11,5,3,6,1,10,4,\ldots,4,1,10,2,0,12,1,0\right\} $ \\ \hline
\end{tabular}
\end{center}
\end{table}

\subsection{Examples}

In this section we present some concrete examples of the theoretical part
established in the previous sections.
\begin{example}
For $p=7,$ and  $r=3$ we have the extended Fibonacci sequence
\begin{equation}
E=\{0,1,1,2,4,0,6,3,2,4,\ldots,0,6,1,0,0,1,1,\ldots\}
\end{equation}
with $l=48$ and $\beta\left( p\right)=12.$ Then by Theorem \ref{t} there exists a cyclic code of parameters  $[48,3,36]_{7}.$
\end{example}
\begin{example}
\label{e3} For $p=11,$ we have $F=\left\{ 0,1,1,2,3,5,8,2,10,1,0,1,\ldots
\right\} $ and hence the period of $F$, $l_{11}$ is $10$ and $\alpha \left(
11\right) =10.$ By the Theorem \ref{rob1}, $10=10.\beta \left( 11\right) ,$
thus, $\beta \left( 11\right) =1.$ Also, from Theorem \ref{siap}, the cyclic
linear code $C,$ generated by $f(x)$ is an MDS code of parameters $%
[10,2,9]_{11}$ and weight polynomial $u^{10}+100uv^{9}+20v^{10}.$ Indeed,
the Fibonacci polynomial of $F$ is $%
f(x)=x+x^{2}+2x^{3}+3x^{4}+5x^{5}+8x^{6}+2x^{7}+10x^{8}+x^{9},$ and the
minimal polynomial of $F$ is $g(x)=\gcd (f(x),x^{10}-1)=\frac{x^{10}-1}{%
x^{2}+x-1}=x^{8}+10x^{7}+2x^{6}+8x^{5}+5x^{4}+3x^{3}+2x^{2}+x+1.$ The
generator matrix of $C$ is the following:
\begin{equation}
G=
\begin{pmatrix}
1 & 0 & 1 & 1 & 2 & 3 & 5 & 8 & 2 & 10 \\
0 & 1 & 1 & 2 & 3 & 5 & 8 & 2 & 10 & 1%
\end{pmatrix}.
\end{equation}
Moreover, by the Theorem \ref{siap4}, the dual code of $C,$ $C^{\bot }$ is
also an MDS code with parameters $[10,8,3]_{11}$ and generator matrix:
\begin{equation}
H=%
\begin{pmatrix}
1 & 0 & 0 & 0 & 0 & 0 & 0 & 0 & 10 & 10 \\
0 & 1 & 0 & 0 & 0 & 0 & 0 & 0 & 10 & 9 \\
0 & 0 & 1 & 0 & 0 & 0 & 0 & 0 & 9 & 8 \\
0 & 0 & 0 & 1 & 0 & 0 & 0 & 0 & 8 & 6 \\
0 & 0 & 0 & 0 & 1 & 0 & 0 & 0 & 6 & 3 \\
0 & 0 & 0 & 0 & 0 & 1 & 0 & 0 & 3 & 9 \\
0 & 0 & 0 & 0 & 0 & 0 & 1 & 0 & 9 & 1 \\
0 & 0 & 0 & 0 & 0 & 0 & 0 & 1 & 1 & 10%
\end{pmatrix}.
\end{equation}
\end{example}
\begin{example}
\label{e4} For $p=19,$ we have
\[
F=\left\{ 0,1,1,2,3,5,8,13,2,15,17,13,11,5,16,2,18,1,0,1,\ldots \right\}
\]
and hence we have $l_{19}=18,\alpha \left( 19\right) =18,$ and by the
Theorem \ref{rob1}, $18=18\beta \left( 19\right) ,$ thus, $\beta \left(
19\right) =1.$ From Theorem \ref{siap}, the cyclic linear code $C,$
generated by $f(x)$ is an MDS code with parameters $[18,2,17]_{19}$ and
weight polynomial $u^{18}+324uv^{17}+36v^{18}.$ Indeed, the Fibonacci
polynomial of $F$ is $%
f(x)=x+x^{2}+2x^{3}+3x^{4}+5x^{5}+8x^{6}+13x^{7}+2x^{8}+15x^{9}+17x^{10}+13x^{11}+11x^{12}+5x^{13}+16x^{14}+2x^{15}+18x^{16}+x^{17},
$ and the minimal polynomial of $F$ is $g(x)=\gcd (f(x),x^{18}-1)=\frac{%
x^{18}-1}{x^{2}+x-1}%
=x^{16}+18x^{15}+2x^{14}+16x^{13}+5x^{12}+11x^{11}+13x^{10}+17x^{9}+15x^{8}+2x^{7}+13x^{6}+8x^{5}+5x^{4}+3x^{3}+2x^{2}+x+1.$
Furthermore, by the Theorem \ref{siap4}, the dual code of $C,$ $C^{\bot }$
is also an MDS code of type $[18,16,3]_{19}$.
\end{example}
\begin{example}
\label{e5} If we take $p=7,$ then $F=\left\{
0,1,1,2,3,5,1,6,0,6,6,5,4,2,6,1,0,1,\ldots \right\} $ and hence we have $%
l=16,$ $\alpha \left( 7\right) =8,$ and by the Theorem \ref{rob1}, $%
16=8\beta \left( 7\right) ,$ so $\beta \left( 7\right) =2.$ By the Theorem %
\ref{siap2} and Corollary \ref{siap3}, $C$ is an optimal code of type $%
[16,2,14]_{7}$ with weight polynomial $u^{16}+48u^{2}v^{14}$. Actually, the
Fibonacci polynomial of $F$ is $%
f(x)=x+x^{2}+2x^{3}+3x^{4}+5x^{5}+x^{6}+6x^{7}+6x^{9}+6x^{10}+5x^{11}+4x^{12}+2x^{13}+6x^{14}+x^{15},
$ and the minimal polynomial of $F$ is $g(x)=\gcd (f(x),x^{16}-1)=\frac{%
x^{16}-1}{x^{2}+x-1}%
=x^{14}+6x^{13}+2x^{12}+4x^{11}+5x^{10}+6x^{9}+6x^{8}+6x^{6}+x^{5}+5x^{4}+3x^{3}+2x^{2}+x+1.$
This code attains the Griesmer bound and hence it is an optimal code.
\end{example}

\section{Fibonacci Codes and Secret Sharing Schemes}

\subsection{Secret Sharing Schemes From Codes}

Secret sharing system is a method of projecting a secret data to finitely
many participants with the aim that a designed number of or designed
participants can recover the data. In this system, a secret data $s$ is
divided into shares and distributed to participants from the set $P=\left\{
P_{1},P_{2},...,P_{n-1}\right\} $ in such a way that only authorized subsets
of $P$ can reconstruct the secret whereas unauthorized subsets cannot
reconstruct the secret. There are several secret sharing system in
literature \cite{Asmuth},\cite{Blakley},\cite{McEliece},\cite{Shamir}. One of them is based on
coding theory. In 1993, Massey has shown that every linear code can be used
to construct the secret sharing scheme \cite{Massey}. Let now us recall the
system given by Massey. Let $C$ be an $\left[ n,k,d\right] $ linear code
over finite field $\mathbb{F}_{p}$ and $G=\left[ g_{0},g_{1},...,g_{n-1}%
\right] $ be a generator matrix of $C$ where $g_{i}$'s are the column
vectors of $G$. In this system, column vectors of $G$ are nonzero. Dealer,
who is a person building the system, randomly choose a vector from $u=\left(
u_{0},u_{1},...,u_{k-1}\right) \in \mathbb{F}_{p}^{k}$ to generate the
codeword $uG=\left( v_{0},v_{1},...,v_{n-1}\right) $. The dealer picks the
first coordinate of a codeword as a secret, i.e., $s=v_{0}=ug_{0}$, and
distributes $v_{i}$ to participants $P_{i}$ as a share for $1\leq i\leq n-1$%
. Since $s=v_{0}=ug_{0}$, it is easily seen that set of shares $\left\{
v_{i_{1}},v_{i_{2}},...,v_{i_{t}}\right\} $ determines the secret $s$ if and
only if\ $g_{0}$ is a linear combination of $%
g_{i_{1}},g_{i_{2}},...,g_{i_{t}}$. To recover the secret $s$, firstly the
linear equation $g_{0}=\sum\limits_{j=1}^{t}x_{j}g_{i_{j}}$ is solved and $%
x_{j}$ is found, then the secret is computed by%
\begin{equation*}
v_{0}=ug_{0}=\sum\limits_{j=1}^{t}x_{j}ug_{i_{j}}=\sum%
\limits_{j=1}^{t}x_{j}v_{i_{j}}\text{.}
\end{equation*}
\begin{definition}
\cite{Massey} Let $v$ be a vector of length $n$ over $\mathbb{F}_{p}.$ The
support of $v$ is defined as%
\[
supp(v)=\left\{ 0\leq i\leq n-1:v_{i}\neq 0\right\} .
\]%
We say that a vector $v_{2}$ covers a vector $v_{1}$ if the support of
vector $v_{2}$ contains that of $v_{1} i.e. supp(v_{1})\subseteq
supp(v_{2}). $
\end{definition}
\begin{definition}
\cite{Massey} A nonzero vector $c$ is called minimal if it only covers its
scalar multiples. If the first component of minimal vector $c$ is $1$, then
the vector $c$ is called minimal codeword.
\end{definition}
\begin{definition}
\cite{Li} The family of all authorized subsets of $P$ is called access
structure of the scheme. Authorized subsets of $P$ are called minimal access
sets if they can reconstruct the secret $s$ but any of its proper subsets
cannot reconstruct the secret $s$.
\end{definition}
Hence we have the following main lemma:
\begin{lemma}
\label{l4}\cite{Massey} Let $C$ be an $\left[ n,k,d\right] $ linear code
over finite field $\mathbb{F}_{p}$ and $C^{\perp }$ be the dual code of $C$.
In the secret sharing scheme based on $C,$ a set of shares $\left\{
v_{i_{1}},v_{i_{2}},...,v_{i_{t}}\right\} $ recovers the secret $s$ if and
only if\ there is a codeword in $C^{\perp }$ such that%
\[
\left( 1,0,...,0,c_{i_{1}},0,...,c_{i_{t}},0,...,0\right)
\]%
where $c_{i_{j}}\neq 0$ for at least one $j,$ $1\leq i_{1}<...<i_{m}\leq n-1$
and $1\leq m\leq n-1.$
\end{lemma}
From Lemma \ref{l4}, it is clear that there is a one to one correspondence
between the set of minimal access sets and sets of minimal codewords. But,
it is very hard to find the minimal codewords of linear codes in general.

\subsection{Access Structures From Fibonacci Codes}

In this section, we consider the secret sharing schemes obtained from
Fibonacci codes whose minimal codewords can be characterized. Let us remind
two lemmas in the literature which state the main results how to determine
the access structure.
\begin{lemma}
\label{l5}\cite{Ash},\cite{Barg} Let $C$ be an $[n,k,d]$ code over $\mathbb{F}_{p}.$
Let $w_{\min },w_{\max }$ be a minimum and maximum nonzero weight\ of $C,$
respectively. If%
\[
\frac{w_{\min }}{w_{\max }}>\frac{p-1}{p}
\]%
then each nonzero codeword of $C$ is a minimal vector.
\end{lemma}
Lemma \ref{l5} states that if the weights of a linear code are close enough
to each other, then each nonzero codeword of the code is minimal. The
following lemma characterizes minimal access set of $C$ where each nonzero
codeword is a minimal vector.
\begin{lemma}
\label{l6}\cite{Ding} Let $C$ be an $[n,k,d]$ code over $\mathbb{F}_{p}$,
and let $G=[g_{0},g_{1},...,g_{n-1}]$ be its generator matrix. If each
nonzero codeword of $C$ is a minimal vector, then in the secret sharing
scheme based on $C^{\perp}$, there are altogether $p^{k-1}$ minimal access
sets. In addition, we have the followings:
\begin{enumerate}
\item If $g_{i}$ is a multiple of $g_{0}$, $1\leq i\leq n-1$, then
participant $P_{i}$ must be in every minimal access set. Such a participant
is called a dictatorial participant.
\item If $g_{i}$ is not a multiple of $g_{0}$, $1\leq i\leq n-1$, then
participant $P_{i}$ must be in $\left( p-1\right) p^{k-2}$ out of $p^{k-1}$
minimal access sets.
\end{enumerate}
\end{lemma}
\begin{theorem}
\label{t10} From Corollary \ref{c7} and Table \ref{tab:3}, in the secret
sharing scheme based on the dual code of the code with parameters $\left[
2p+2,2,2p\right] $ over $\mathbb{F}_{p},$ there are $p$ minimal access sets
and $P_{p+1}$ is a dictatorial participant. Furthermore, each of the other
participant $P_{i}$ is involved in $\left( p-1\right) $ minimal access sets.
\end{theorem}
\begin{proof}
The codes of parameters $\left[ 2p+2,2,2p\right] $ are one
weight codes and by Lemma \ref{l5} all codewords of such codes are minimal.
Thus, from Lemma \ref{l6} there are $p^{2-1}=p$ minimal access sets. Also,
from the Fibonacci sequence $\mod p$ and $\beta \left( p\right) =2,$ we have
two multiple parts in a sequence. So, the $g_{p+1}$ column is a multiple of $%
g_{0}.$ This means that $P_{p+1}$ is a dictatorial participant. The other
remaining participants are involved in $\left( p-1\right) p^{2-2}=\left(
p-1\right) $ minimal access sets.
\end{proof}
\begin{example}
The code $C$ given in Example \ref{e5} has parameters $\left[ 16,2,14\right]
$ with the following weight distribution:
\[
u^{16}+48u^{2}v^{14}.
\]%
In the secret sharing scheme based on the dual code of $C,$ the number of
minimal access sets are $7,$ and the list of all these minimal access sets
are as follows:
\begin{eqnarray*}
&&\left\{ 2,3,4,5,6,7,8,10,11,12,13,14,15\right\} ,\left\{
1,2,3,4,5,6,8,9,10,11,12,13,14\right\} , \\
&&\left\{ 1,2,3,4,5,7,8,9,10,11,12,13,15\right\} ,\left\{
1,2,4,5,6,7,8,9,10,12,13,14,15\right\} , \\
&&\left\{ 1,2,3,5,6,7,8,9,10,11,13,14,15\right\} ,\left\{
1,2,3,4,6,7,8,9,10,11,12,14,15\right\} , \\
&&\text{ \ \ \ \ \ \ \ \ \ \ \ \ \ \ \ \ \ \ \ \ \ \ \ \ }\left\{
1,3,4,5,6,7,8,9,11,12,13,14,15\right\} .
\end{eqnarray*}%
where $\{2,3,4,5,6,7,8,10,11,12,13,14,15\}$ denotes the access set
\[
\{P_{2},P_{3},P_{4},P_{5},P_{6},P_{7},P_{8},P_{10},P_{11},P_{12},P_{13},P_{14},P_{15}\}.
\]%
In this example, $P_{8}$ is dictatorial participant and each participant is
involved in exactly $6$ minimal access sets.
\end{example}
\begin{theorem}
\label{t11} From Corollary \ref{c7} and Table \ref{tab:2}, in the secret
sharing schemes based on the dual code of the code of parameters $\left[
p-1,2,p-2\right] $ over $\mathbb{F}_{p}$, there are $p-2$ minimal access
sets. Furthermore, each participant $P_{i}$ is involved in $\left(
p-3\right) $ minimal access sets.
\end{theorem}
\begin{proof}
From Table \ref{tab:2}, it is easily seen that the codewords\
of weight $p-2$ whose first component $1$ are minimal codewords. We have in
total $\frac{p^{2}}{p}=p$ codewords whose first component are $1$. We should
remove full weight codes whose first component are $1,$ to get minimality.
There are $2\left( p-1\right) $ codewords which have weight $p-1.$ So, there
are $\frac{2\left( p-1\right) }{p-1}=2$ codewords whose first component is $%
1 $. This give us, the total number of the codewords of weight\ $p-2$ whose
first component is $1.$ Consequently, we have $p-2$ minimal codewords.
Therefore, there are $p-2$ minimal access sets.
\newline
\textbf{Second part of proof:} If we fix $i,$ $1\leq i\leq p-2,$ participant
$P_{i}$ is involved in $p-2$ minimal access sets.\ In this case we can
choose remaining participants to recover the key in $\left(
\begin{array}{c}
p-3 \\
p-4%
\end{array}%
\right) =p-3$ ways. This is a contradiction to $p-2$ minimal access sets.
Thus each of participant $P_{i}$ is involved in $\left( p-3\right) $ minimal
access sets.
\end{proof}
We only state and skip the proofs of the following two lemmas since they can
be proved similarly as in Theorem \ref{t10} and Theorem \ref{t11}.
\begin{corollary}
\label{c8} From Corollary \ref{c7} and Table \ref{tab:3}, in the secret
sharing scheme based on the dual code of the code of parameters $\left[
p-1,2,p-3\right] $ over $\mathbb{F}_{p}$ , there are $\frac{p-3}{2}$ minimal
access sets and $P_{\frac{_{p-1}}{2}}$ is a dictatorial participant.\
Furthermore, each of the other participant $P_{i}$ is involved in $\left(
\frac{p-5}{2}\right) $ minimal access sets.
\end{corollary}
\begin{corollary}
\label{c9} From Corollary \ref{c7} and Table \ref{tab:4}, in the secret
sharing schemes based on the dual code of the code of parameters $\left[
p-1,2,p-5\right] $ over $\mathbb{F}_{p}$, there are $\frac{p-5}{4}$ minimal
access sets and the set of $\left\{ P_{\frac{_{p-1}}{4}},P_{\frac{p-1}{2}%
},P_{\frac{3(p-1)}{4}}\right\} $ are a dictatorial participants.\
Furthermore, each of the other participant $P_{i}$ is involved in $\left(
\frac{p-9}{4}\right) $ minimal access sets.
\end{corollary}
\begin{corollary}
\label{c10} From Corollary \ref{c7} and Table \ref{tab:4}, in the secret
sharing schemes based on the dual code of the code of parameters $\left[
2p+2,2,2p-2\right] $ over $\mathbb{F}_{p}$, there are $\frac{p-1}{2}$minimal
access sets and $\left\{P_{\frac{_{p+1}}{2}},P_{p+1},P_{\frac{3(p+1)}{2}%
}\right\} $ are a dictatorial participant. Furthermore, each of the other
participant $P_{i}$ is involved in $\left(\frac{p-3}{2}\right)$ minimal
access sets.
\end{corollary}
\begin{example}
For $\beta \left( p\right) =1$ and $l=p-1,$ we have a linear code of
parameters $[10,2,9]_{11}$ given in Example \ref{e3}. From Theorem \ref{t11}, there are $p-2=9$ minimal access sets and the list of all these minimal
access sets are as follows:
\begin{eqnarray*}
&&\left\{ 2,3,4,5,6,7,8,9\right\} ,\left\{ 1,2,3,4,5,6,7,8\right\} ,\left\{
1,2,3,4,5,6,7,9\right\} , \\
&&\left\{ 1,2,3,5,6,7,8,9\right\} ,\left\{ 1,2,4,5,6,7,8,9\right\} ,\left\{
1,2,3,4,6,7,8,9\right\} , \\
&&\left\{ 1,2,3,4,5,6,8,9\right\} ,\{1,2,3,4,5,7,8,9\},\{1,3,4,5,6,7,8,9\}.
\end{eqnarray*}%
In this example, each participant is involved in exactly $p-3=8$ minimal
access sets.
\end{example}
\begin{example}
Let $\beta \left( p\right) =4$\ and $p=13$. By the Theorem \ref{siap2}, we
have linear code of parameters $[28,2,24]_{13}.$ The weight polynomial of the given code is $u^{28}+84u^{4}v^{24}+84v^{28}.$
From Corollary \ref{c10}, there are $\frac{p-1}{2}=6$ minimal access sets
and $P_{7},P_{14},P_{21}$ are dictatorial participants. The list of all
these minimal access sets are as follows:%
\begin{eqnarray*}
&&\left\{2,3,4,5,6,7,9,10,11,12,13,14,16,17,18,19,20,21,23,24,25,26,27\right\}, \\
&&\left\{1,2,3,4,5,7,8,9,10,11,12,14,15,16,17,18,19,21,22,23,24,25,26\right\}, \\
&&\left\{1,2,3,4,6,7,8,9,10,11,13,14,15,16,17,18,20,21,22,23,24,25,27\right\}, \\
&&\left\{1,2,4,5,6,7,8,9,11,12,13,14,15,16,18,19,20,21,22,23,25,26,27\right\}, \\
&&\left\{1,2,3,5,6,7,8,9,10,12,13,14,15,16,17,19,20,21,22,23,24,26,27\right\}, \\
&&\left\{1,3,4,5,6,7,8,10,11,12,13,14,15,17,18,19,20,21,22,24,25,26,27\right\}.
\end{eqnarray*}%
\end{example}
\section{Conclusion}

In this paper, we studied cyclic codes with generator polynomials derived
from Fibonacci sequences modulo a prime $p.$ We showed that such cyclic
codes enjoy very good properties as they give examples of MDS and optimal
cyclic codes. Also, we are able to determine all parameters of such codes.
Finally, we present applications to secret sharing schemes via Fibonacci
codes.
\begin{acknowledgements}
This research is partially supported by TUBITAK-ARDEB under the project with
grant number 114F388
\end{acknowledgements}

\end{document}